\documentclass[12pt,reqno]{amsart}
\makeatletter
\def\l@subsection{\@tocline{2}{0pt}{2.5pc}{5pc}{}}
\def\l@subsubsection{\@tocline{2}{0pt}{5pc}{7.5pc}{}}
\makeatother
\usepackage{frcursive}
\usepackage{yfonts}
\usepackage{graphicx}
\usepackage{empheq}
\usepackage{mathtools}
\usepackage{tikz}
\usepackage{amsbsy}
\usepackage{avant}
\usepackage{amsmath}
\usepackage{amsthm}
\usepackage{amssymb}
\usepackage{mathrsfs}
\usepackage{amsfonts}
\usepackage{newlfont}
\usepackage[sans]{dsfont}
\usepackage{dcolumn}
\usepackage{bm}
\usepackage{bbm}
\usepackage{stmaryrd}
\usepackage{bbm}
\usepackage{relsize}
\usepackage{euscript}
\usepackage{eufrak}
\usepackage{enumerate}
\usepackage[margin=0.8in]{geometry}
\numberwithin{equation}{section}
\newtheorem{thm}{Theorem}[section]

\newtheorem{cor}[thm]{Corollary}
\newtheorem{lem}[thm]{Lemma}
\newtheorem{prop}[thm]{Proposition}
\newtheorem{defn}[thm]{Definition}

\begin{document}
\allowdisplaybreaks{
\title[]{A STOCHASTIC DIFFERENTIAL EQUATION FOR LASER PROPAGATION IN MEDIAS WITH RANDOM GAUSSIAN ABSORPTION COEFFICIENTS:~~MODIFIED BEER'S LAW
SOLUTION VIA A VAN KAMPEN CLUSTER EXPANSION}
\author{Steven D Miller}\email{stevendm@ed-alumni.net}
\address{}
\maketitle
\begin{abstract}
Let $\mathbb{I\!D}=[0,\mathrm{L}]\subset\mathbb{R}^{+}$ be a slab geometry with boundaries $z=0$ and $z=\mathrm{L}$. A laser beam with a flat collimated incident intensity $\psi_{o}$ enters the slab along the z-axis or unit vector $\widehat{\mathbf{e}}_{3}$ at $z=0$. The slab contains purely absorbing matter with an absorption coefficient of $\mathsf{A}$ (in $cm^{-1}$)with respect to the wavelength. If $\mathsf{A}$ is constant and homogenous then the beam decays as Beer's law $\psi(z,\widehat{\mathbf{e}}_{3})=\psi_{o}\exp(-\mathsf{A} z)$. The absorption coefficient is randomly fluctuating in space as $\mathbf{A}(z)=\mathsf{A}+\alpha \mathsf{A}{\mathbf{G}}(z)$, where $\alpha>0$ determines the magnitude of the fluctuations, and the Gaussian random function/noise $\mathbf{G}(z)$ has correlation $\mathbb{E}\big\lbrace
\mathlarger{\mathbf{G}}(z_{1}){\otimes}~\mathlarger{\mathbf{G}}(z_{2})\big\rbrace=\phi(z_{1},z_{2};\xi)=\bm{\mathsf{C}}\exp(-|z_{1}-z_{2}|^{2}\xi^{-2})$ with correlation length $\xi$. The beam propagation and attentuation within the medium is then described by the stochastic differential equation
\begin{align}
d\widehat{\bm{\psi}(z,\mathbf{e}_{3})}=-\mathsf{A}\widehat{\bm{\psi}(z,\mathbf{e}_{3})}dz-\alpha\mathsf{A}\widehat{\bm{\psi}(z,\mathbf{e}_{3})}\nonumber
\mathbf{G}(z)dz
\end{align}
The stochastic average of the solution is derived via a Van-Kampen-like cumulant expansion, truncated at second order for Gaussian dominance, giving a modified Beer's law of the form
\begin{align}
&\mathbb{E}\big\lbrace\widehat{\bm{\psi}(z,\mathbf{e}_{3})}\big\rbrace
=\psi_{o}\exp(-\mathsf{A}z)\exp\left(\frac{1}{2}\alpha^{2}\mathsf{A}^{2}\bm{\mathsf{C}}\int_{0}^{z}\bigg|\int_{0}^{z_{1}}\mathlarger{\phi}(z_{1},z_{2};\xi))dz_{2}\bigg|dz_{1}\right)
\nonumber\\&=\psi_{o}\exp(-\mathsf{A}z)\exp\left(\frac{1}{4}\alpha^{2}\mathsf{A}^{2}\bm{\mathsf{C}}\xi\left[\exp(-z^{2}/\xi^{2})\left(\sqrt{\pi}z Erf\left(\frac{z}{\xi.}\right)\exp\left(\frac{z^{2}}{\xi^{2}}\right)+\xi\right)-\xi\right]\right)\nonumber
\end{align}
The deterministic Beer's law is recovered as $\alpha\rightarrow 0$ and for $z\gg\xi$.
\end{abstract}
\raggedbottom
\maketitle
\section{{Introduction}}
In this paper, a stochastic differential equation is formulated and solved, to describe laser propagation into absorbing matter where the absorption coefficient is a random Gaussian function; that is, the absorption coefficient is a randomly fluctuating spatial function, which fluctuates about a mean value and has Gaussian statistics. We first consider the fundamental description of laser propagation in matter where the absorption and scatter coefficients are constant and homogenous, in terms of the radiative transfer (RT) of photons which are scattered and absorbed by the matter. The fluence or flux evolving via radiative transfer equation (RTE). This is the fundamental equation of 'radiation hydrodynamics' and has found very many applications in physics and engineering\textbf{[1-12].}
\begin{enumerate}
\item In astrophysics to describe radiative transfer in stars; in nuclear engineering to describe neutron transport in reactors; in laser-matter physics, to describe the evolution of an initially collimated laser beam into a scattered distribution of light and energy within matter, and the propagation of scatted light in turbid solids, liquids and medias. Radiation hydrodynamics is also necessary to compute the propagation of laser energy into the fuel capsule during laser-induced fusion or inertial confinement experiments.\textbf{[1-11]}.
\item In biomedical optics, the RTE and its numerical and diffusion approximation have been extensively applied to the description of photon transport and scattering within biological tissues, whole blood and blood cell suspensions; in the development of laser scalpels and laser surgery, where the penetration of laser energy into the tissue determines how deeply the laser scalpel will cut; in photodynamic cancer therapy and cancer radiation dosimetry \textbf{[12-21]}.
\end{enumerate}
The fundamental radiative transport equation (RTE) is related to the Boltzmann equation of a gas from kinetic theory, and describes the evolution of
a flux or intensity $\psi(\mathbf{x},t,\widehat{\bm{\omega}})$, from initial data and with suitable boundary conditions on $\partial\mathbb{I\!D}$. The RTE states that an initial 'beam' or flux of photons loses energy through scattering away from the beam and through absorption/capture, and gains energy from photon sources or emissivity within the medium and scattering directed towards the beam. In effect, it is an energy balance law and given by the integro-differential equation. Details of radiative transfer and the RTE and its many applications are given in many classic and modern texts and there is a very extensive literature.\textbf{[1-12}, and references therein].
\subsection{\textbf{The radiative transport equation and reduction to Beer's law of attenuation for zero scattering}}
Let $\mathbb{I\!D}\subset\mathbb{R}^{3}$ be a domain containing matter which either absorbs and scatters photons. The number density of scattering-absorbing particles is at least greater than 1 percent. The scattering and absorption coefficients are $\mathsf{S}$ and $\mathsf{A}$ with $\mathsf{T}=\mathsf{S}+\mathsf{A}$ and have units of $cm^{-1}$. The mean free path to scatter or absorption is then $\mathsf{M}=\mathsf{T}^{-1}$. The \textbf{radiance} or \textbf{luminosity} in direction $\widehat{\bm{\omega}}$ is $\psi(\mathbf{x},t,\widehat{\bm{\omega}})$, where $(\mathbf{x},t)\in\mathbb{I\!D}\times[0,\infty)$, and represents the flow of energy per unit area per unit solid angle per unit time. The vector $\widehat{\bm{\omega}}$ is a unit vector and $d\widehat{\bm{\omega}}$ is a solid angle element. If a pencil or collimated laser beam of initial intensity $\psi_{0}$ enters scattering-absorbing matter contained in a domain $\mathbb{I\!D}$ on its boundary $\partial\mathbb{I\!D}$, then the fundamental radiative transport equation (RTE) is \textbf{[1-9,12]}.
\begin{align}
&\frac{\partial \psi(\mathbf{x},t,\widehat{\bm{\omega}})}{\partial t}+\widehat{\bm{\omega}}.\nabla\psi(\mathbf{x},t,\widehat{\bm{\omega}})
=-\big(\mathsf{A}+\mathsf{S}){\psi(\mathbf{x},t,\widehat{\bm{\omega}})}\nonumber\\&+
\frac{\mathsf{S}}{4 \pi}\mathlarger{\int}_{4\pi}\mathlarger{\mathscr{P}}(\widehat{\bm{\omega}}\rightarrow\widehat{\bm{\omega}}^{\prime})
{\psi(\mathbf{x},t,\widehat{\bm{\omega}}^{\prime})}+\mathlarger{\mathscr{E}}(\mathbf{x},t,\widehat{\bm{\omega}}),~~~(\mathbf{x},t)\in\mathbb{I\!D}\times\mathbb{R}^{+}
\end{align}
The initial Cauchy data is $\psi(\mathbf{x},0,\widehat{\bm{\omega}})=0$ within the matter, and with suitable boundary conditions for the incident collimated laser beam $\psi_{c}(\mathbf{x}\in\partial\mathbb{I\!D},t,\hat{\mathbf{n}})$ on $\partial\mathbb{I\!D}$, where $\hat{\mathbf{n}}$ is a unit normal vector pointing into the domain.
The term $\mathscr{E}(\mathbf{x},t,\widehat{\bm{\omega}})$ represents emissivity of the matter, which will be taken here to be zero.

If the total intensity or radiance within the matter is a linear combination of collimated and scattered laser energy then
\begin{align}
\psi(\mathbf{x},t,\widehat{\bm{\omega}})=\psi_{c}(\mathbf{x},t,\widehat{\bm{\omega}})+\psi_{s}(\mathbf{x},t,\widehat{\bm{\omega}})
\end{align}
so the RTE becomes
\begin{align}
&\frac{\partial}{\partial t}\big(\psi_{c}(\mathbf{x},t,\widehat{\bm{\omega}})+\psi_{s}(\mathbf{x},t,\widehat{\bm{\omega}})\big)
+\widehat{\bm{\omega}}.\nabla\big(\psi_{c}(\mathbf{x},t,\widehat{\bm{\omega}})+\psi_{s}(\mathbf{x},t,\widehat{\bm{\omega}})\big)\nonumber\\&
=-\big(\mathsf{A}+\mathsf{S})\big(\psi_{c}(\mathbf{x},t,\widehat{\bm{\omega}})+\psi_{s}(\mathbf{x},t,\widehat{\bm{\omega}})\big)\nonumber\\&+
\frac{\mathsf{S}}{4 \pi}\mathlarger{\int}_{4\pi}\mathlarger{\mathscr{P}}(\widehat{\bm{\omega}}\rightarrow\widehat{\bm{\omega}}^{\prime})
\big(\psi_{c}(\mathbf{x},t,\widehat{\bm{\omega}}^{\prime})+\psi_{s}(\mathbf{x},t,\widehat{\bm{\omega}}^{\prime})\big),(\mathbf{x},t)\in\mathbb{I\!D}\times\mathbb{R}^{+}
\end{align}
In the steady state $\psi(\mathbf{x},t,\widehat{\bm{\omega}})=\psi(\mathbf{x},\widehat{\bm{\omega}})$ and the RTE is then
\begin{align}
&\widehat{\bm{\omega}}.\nabla{\psi(\mathbf{x},\widehat{\bm{\omega}})}
=-\big(\mathsf{A}+\mathsf{S}\big)\psi_{c}(\mathbf{x},\widehat{\bm{\omega}})+\psi_{s}(\mathbf{x},\widehat{\bm{\omega}})\big)\nonumber\\&+
\frac{\mathsf{S}}{4 \pi}\mathlarger{\int}_{4\pi}\mathlarger{\mathscr{P}}(\widehat{\bm{\omega}}\rightarrow\widehat{\bm{\omega}}^{\prime})
\big(\psi_{c}(\mathbf{x},\widehat{\bm{\omega}}^{\prime})+\psi_{s}(\mathbf{x},\widehat{\bm{\omega}}^{\prime})\big),~~~\mathbf{x}\in\mathbb{I\!D}
\end{align}
This would describe the evolution of a steady state continuous source laser beam entering the matter. Here, the function $\mathcal{P}(\widehat{\bm{\
\omega}}\rightarrow\widehat{\bm{\omega}}^{\prime})$ is the scattering phase function, which gives the probability that
photons propagating along direction $\widehat{\bm{\omega}}$ are scattered into $\widehat{\bm{\omega}}^{\prime}$. It has a Legendre polynomial series representation $\textbf{[1-9,12]}$
\begin{align}
\mathlarger{\mathscr{P}}(\widehat{\bm{\omega}}\rightarrow\widehat{\bm{\omega}}^{\prime})
=\frac{1}{4\pi}\sum_{n=0}^{\infty}a_{n}\mathrm{P}_{n}(\widehat{\bm{\omega}}\rightarrow\widehat{\bm{\omega}}^{\prime})
\end{align}
The anisotropy of the scattering is given by the factor $\bm{\mathsf{g}}\in[0,1]$ where
\begin{align}
\bm{\mathsf{g}}=\int_{4\pi}(\widehat{\bm{\omega}}.\widehat{\bm{\omega}}^{\prime})\mathlarger{\mathscr{P}}(\widehat{\bm{\omega}}\rightarrow\widehat{\bm{\omega}}^{\prime})d\bm{\omega}
\end{align}
The RTE is an integro-differential equation and as such is very difficult to solve. Numerical analysis, Monte Carlo methods and diffusion and $\mathrm{P}_{N}$ approximations have been extensively applied, and again the literature is by now quite vast. (\textbf{[1-12]}, and references therein.)

If the scattering coefficient is very small with $\mathsf{A}\sim 0$, or set to zero $\mathsf{A}=0$, and also $\mathscr{P}(\widehat{\bm{\omega}}\rightarrow\widehat{\bm{\omega}}^{\prime})=0$, then the RTE simplifies considerably and reduces to the linear transport PDE
\begin{align}
&\frac{\partial\psi_{c}(\mathbf{x},t,\widehat{\bm{\omega}})}{\partial t}+\widehat{\bm{\omega}}.\nabla\psi_{c}(\mathbf{x},t,\widehat{\bm{\omega}})=-\mathsf{A}\psi_{c}(\mathbf{x},t,\widehat{\bm{\omega}}),
~~~(\mathbf{x},t)\in\mathbb{I\!D}\times\mathbb{R}^{+}
\end{align}
with $\psi_{s}(\mathbf{x},t,\widehat{\bm{\omega}})=0$. For a plane or slab geometry of thickness $\mathrm{L}$ , we can choose $\mathbf{x}=(0,0,z)$ and $\widehat{\bm{\omega}}=(0,0,\hat{\mathbf{e}}_{3})$ then $\mathbb{I\!D}=[0,\mathrm{L}]$ for a finite slab and $\mathbb{I\!D}=[0,\infty)=\mathbb{R}^{+}$ for a semi-infinite slab, with the incident beam along $\widehat{\mathbf{e}}_{3}$. Then
\begin{align}
&\frac{\partial\psi_{c}(z,t,\widehat{\mathbf{e}}_{3})}{\partial t}
+\frac{\partial\psi_{c}(z,t,\widehat{\mathbf{e}}_{3})}{\partial z}=-\mathsf{A}\psi_{c}(z,t,\widehat{\mathbf{e}}_{3}),~~~~~~(z,t)\in\mathbb{I\!D}\times\mathbb{R}^{+}
\end{align}
In the steady state $\psi_{c}(z,t,\widehat{\mathbf{e}}_{3})=\psi_{c}(\mathrm{z},\widehat{\mathbf{e}}_{3})$
\begin{align}
&\frac{d\psi_{c}(z,\widehat{\mathbf{e}}_{3})}{d z}=-\mathsf{A}\psi_{c}(z,\widehat{\mathbf{e}}_{3}),~~~~~~z\in\mathbb{I\!D}
\end{align}
\begin{figure}[htb]
\begin{center}
\includegraphics[height=2.0in,width=6.0in]{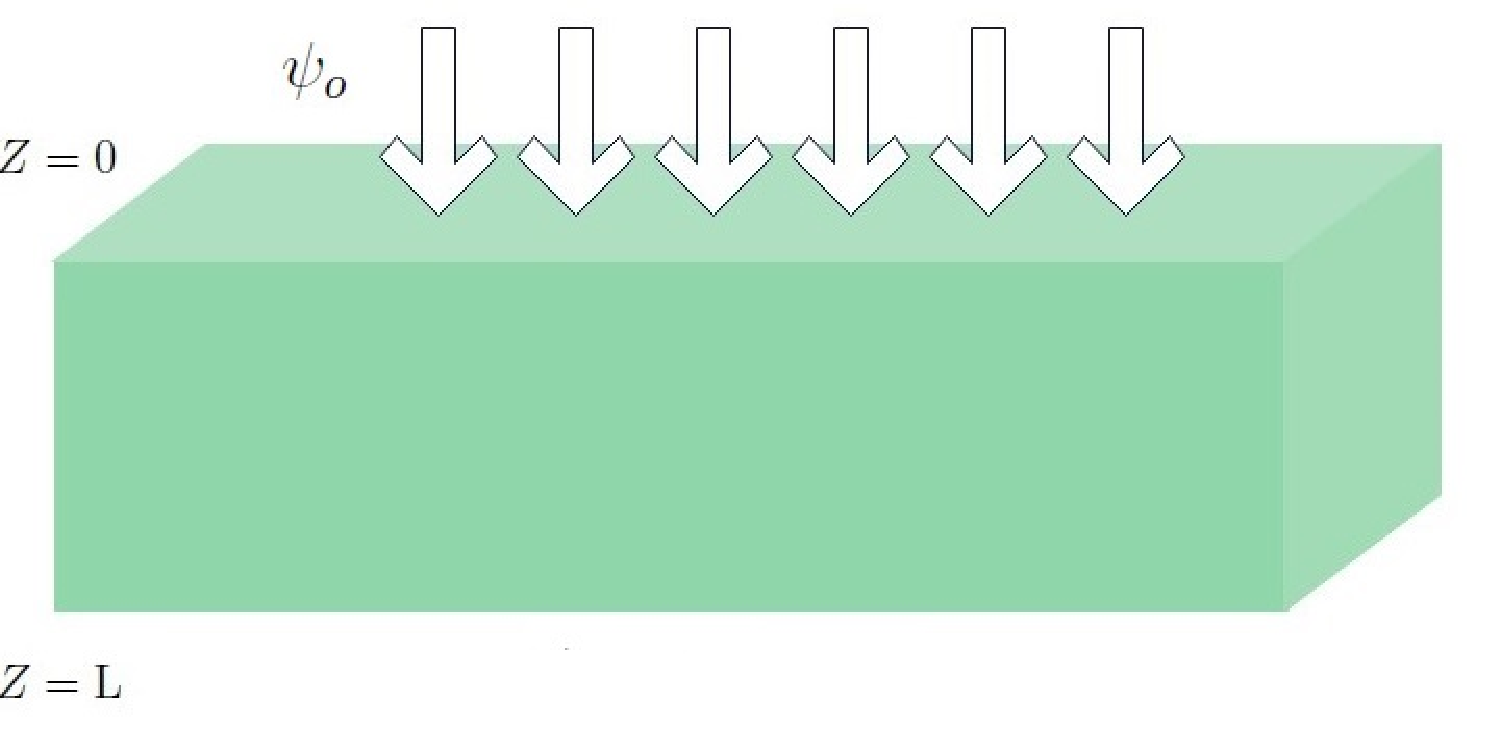}
\caption{ A monochromatic and collimated flat laser beam of intensity $\psi_{o}$ incident upon a slab of matter of thickness $z=\mathrm{L}$ and absorption coefficient $\mathsf{A}$.}
\end{center}
\end{figure}
For a flat collimated laser beam along the z-axis, normal to the slab [Figure 1], the incident intensity or flux is $\psi_{c}(z=0,\widehat{\mathbf{e}}_{3})=\psi_{o}$.
This ODE is easily solved
\begin{align}
&\mathlarger{\int}_{\psi_{o}}^{\psi_{c}(z,t,\widehat{\mathbf{e}}_{3})}\frac{d\overline{\psi_{c}(z,t,\widehat{\mathbf{e}}_{3})}}{\overline{\psi_{c}(z,t,\widehat{\mathbf{e}}_{3})}}=
-\mathsf{A}\int_{0}^{z} dz,~~~~~~z\in\mathbb{I\!D}
\end{align}
so that
\begin{align}
\log\left(\frac{\psi_{c}(z,t,\widehat{\mathbf{e}}_{3})}{\psi_{o}}\right)=-\mathsf{A}z
\end{align}
This then gives Beer's law of exponential attenuation or decay of a collimated beam due to absorption.
\begin{align}
\psi_{c}(z,\widehat{\mathbf{e}}_{3})=\psi_{o}\exp(-\mathsf{A}~z)
\end{align}
The flux or intensity at any depth $z$ is then very easily found if $\mathsf{A}$ and $\psi_{o}$ are known. If $\mathsf{S}$ is nonzero but very small then (1.9) becomes
\begin{align}
&\frac{d\psi_{c}(z,\widehat{\mathbf{e}}_{3})}{d z}=-(\mathsf{A}+\mathsf{S})\psi_{c}(z,\widehat{\mathbf{e}}_{3}),~~~~~~z\in\mathbb{I\!D}
\end{align}
giving the Beer-Lambert attenuation law
\begin{align}
\psi_{c}(z,\widehat{\mathbf{e}}_{3})=\psi_{o}\exp(-(\mathsf{A}+\mathsf{S})z),~~~~~~z\in\mathbb{I\!D}
\end{align}
For the purely absorbing case, suppose the absorption coefficient varies with depth so that $\mathsf{A}=\mathsf{A}(z)$ then
\begin{align}
&\frac{d\psi_{c}(z,\widehat{\mathbf{e}}_{3})}{d z}=-\mathsf{A}(z)\psi_{c}(z,\widehat{\mathbf{e}}_{3}),~~~~~~z\in\mathbb{I\!D}
\end{align}
Then
\begin{align}
&\mathlarger{\int}_{\psi_{o}}^{\psi_{c}(z,t,\widehat{\mathbf{e}}_{3})}\frac{d\overline{\psi_{c}(z,t,\widehat{\mathbf{e}}_{3})}}{\overline{\psi_{c}(z,t,\widehat{\mathbf{e}}_{3})}}
=-\int_{0}^{z}\mathsf{A}(z)
 dz,~~~~~~z\in\mathbb{I\!D}
\end{align}
with solution
\begin{align}
\psi_{c}(z,\widehat{\mathbf{e}}_{3})=\psi_{o}\exp\left(-\int_{0}^{z}\mathsf{A}(z)dz\right),~~~~~~z\in\mathbb{I\!D}
\end{align}
\section{\textbf{Stochastic media with a random Gaussian absorption coefficient}}
Suppose now, the matter is considered as a random or stochastic medium with a randomly fluctuating absorption coefficient. This is more realistic for matter comprised of
stochastic mixtures since absorption properties of the matter will not be uniformly homogeneous throughout; liquid media, particle suspensions subjected to thermal noise; for medical laser applications, living tissue which is 'physiologically noisy' due to random variations in temperature, blood flow, chemical composition, oxygenation etc.
Propagation of beams, waves and radio waves in turbulent medias with random refractive index have been extensively studied as well as wave and beam propagation in a
turbulent atmosphere \textbf{[22-25]}. A stochastic scatter coefficient can also be considered, but we will consider only absorption by matter. The random absorption coefficient can be interpreted as a stochastic function or noise, specifically a Gaussian random function or field. Gaussian random fields or functions (GRFs) are defined and presented in more detail in Appendix A. These GRF can be spatio-temporal or just purely spatial. Only purely spatial GRFs are considered.

Classical random fields or functions are well suited to describe structures and properties of systems, that vary randomly in space and/or time. They have found a myriad of useful applications in applied and computational mathematics and science: in the statistical theory of turbulence, medical science, geoscience, engineering, imaging, computer graphics, statistical mechanics and statistics, biology and cosmology $\mathbf{[26-45]}$. Gaussian random fields (GRFs) are of special significance as they can occur spontaneously via the central limit theorem, in systems with a large number of degrees of freedom. GRFs also arise in the study of complex systems like spin glasses, in optimization problems and protein folding \textbf{[43-45]}. Coupling random fields or noise to ODEs or PDEs is also a useful methodology for studying turbulence, random systems, chaos, pattern formation etc. Within pure and applied mathematics, stochastic partial differential equations (SPDEs), arising from the coupling of random fields/noises to PDEs is an evolving research field $\textbf{[46-49]}$. SPDES can model the propagation of heat, diffusions and waves in random medias or turbulent medias. Many dynamical systems are also affected or influenced by noise, which can either destabilize or stabilize a dynamical system \textbf{[51-54]}. All physical systems invariably possess noise on some scale, either of thermal, mechanical or even quantum origin.
\begin{defn}\textbf{(Gaussian random field in 1D)}\newline
Let $\mathbb{I\!D}$ be a slab of depth $\mathrm{L}$, or semi-infinite slab containing matter with (mean) absorption coefficient $\mathsf{A}$. Let $\mathbf{G}(z)$ be a an intrinsic homogenous and isotropic GRF or 'noise' existing and defined for all $z\in\mathbb{I\!D}$. (See Appendix A.) The GRF is determined entirely by its first two moments and has the properties
\begin{enumerate}
\item
The 1st and 2nd moments are
\begin{align}
&\mathbb{M}_{1}(z)=\mathbb{E}\big\lbrace{\mathbf{G}}(z)\big\rbrace=0\\&
\mathbb{M}_{2}(z,z^{\prime})=\mathbb{E}\big\lbrace{\mathbf{G}}(z)\otimes{\mathbf{G}}(z^{\prime})\big\rbrace
=\phi(z,z^{\prime};\xi),~~~~z,z^{\prime}\in\mathbb{I\!D}
\end{align}
The field is regulated in that
\begin{align}
\mathbb{E}\big\lbrace{\mathbf{G}}(z)\otimes{\mathbf{G}}(z)\big\rangle=\phi(0,0^{\prime};\xi)<\infty
\end{align}
so that the noise is non-white and also differentiable and the derivative $\tfrac{d}{dz}\mathbf{G}(z)$ exists and can be defined.
\item
We can choose a Gaussian-decay correlation ansatz of the form
\begin{align}
\mathbf{M}_{2}(z,z^{\prime})=\mathbb{E}\big\lbrace{\mathbf{G}}(z)\otimes{\mathbf{G}}(z^{\prime})\big\rbrace
=\phi(z,z^{\prime};\xi)=\bm{\mathsf{C}}\exp(-|z-z^{\prime}|^{\kappa}\xi^{-\kappa})
\end{align}
where $\bm{\mathsf{C}}$ is dimensionless and greater than zero. For $\kappa=1$ this 'colored noise' and for $\kappa=2$ $\mathbf{G}(z)=\mathbf{B}(z)$
is a Bargman-Fock random field. For $\kappa>2$ then $\mathbf{G}(z)=\mathbf{S}(z)$ and this is a 'super-Gaussian' correlation. The correlation length is $\xi$ so that (2.4) decays rapidly for $|z-z^{\prime}|>\xi$. In the limit that $\xi\rightarrow 0$ then (2.4) reduces to a white noise in space
\begin{align}
&\lim_{\xi\rightarrow 0}\mathbb{E}\left\lbrace{\mathbf{G}}(z)\otimes{\mathbf{G}}(z^{\prime})\right\rbrace
=\lim_{\xi\rightarrow 0}\phi(z,z^{\prime};\xi)\nonumber\\&
=\mathbb{E}\big\lbrace{\mathbf{W}}(z)\otimes{\mathbf{W}}(z^{\prime})\big\rbrace=\bm{\mathsf{C}}\delta(z-z^{\prime})
\end{align}
\item The field or function ${\mathbf{G}}(z)$ is integrable in that
\begin{align}
\int_{\mathbb{I\!D}}{\mathbf{G}}(z)dz\equiv\int_{0}^{z}\mathbf{G}(z)dz
\end{align}
exists within $\mathbb{I\!D}$ and is also a GRF.(See Appendix B)
\item The m-point moments for a set $\lbrace z_{1},...,z_{m}\rbrace\in\mathbb{I\!D}$ are
\begin{align}
\mathbb{E}\big\lbrace{\mathbf{G}}(z_{1})\otimes...\otimes{\mathbf{G}}(z_{m})\big\rbrace=\mathbb{E}\left\lbrace\prod_{q=1}^{m}\otimes{\mathbf{G}}(z_{q})\right\rbrace
\end{align}
and the cumulants are
\begin{align}
\mathbb{K}\big\lbrace{\mathbf{G}}(z_{1})\otimes...\otimes{\mathbf{G}}(z_{m})\big\rbrace=\mathbb{K}\left\lbrace\prod_{q=1}^{m}\otimes{\mathbf{G}}(z_{q})\right\rbrace
\end{align}
where for $m=2$
\begin{align}
\mathbb{K}\big\lbrace{\mathbf{G}}(z_{1})\otimes{\mathbf{G}}(z_{2})\big\rbrace
=\mathbb{E}\big\lbrace{\mathbf{G}}(z_{1})\otimes{\mathbf{G}}(z_{2})\big\rbrace-
\mathbb{E}\big\lbrace{\mathbf{G}}(z_{1})\big\rbrace\mathbb{E}\big\lbrace{\mathbf{G}}(z_{2})\big\rbrace
\end{align}
and so on.
\item For a Gaussian process or function $\mathbf{G}(z)$, all cumulants beyond order $m=2$ are zero so that
\begin{align}
\mathbb{K}\big\lbrace{\mathbf{G}}(z_{1})\otimes...\otimes{\mathbf{G}}(z_{m})\big\rbrace
=\mathbb{K}\left\lbrace\prod_{q=1}^{m}\otimes{\mathbf{G}}(z_{q})\right\rbrace_{m\ge 3}=0
\end{align}
\end{enumerate}
\end{defn}
The problem of spatial and temporal randomness was initially recognised in connection with problems concerning the propagation of electromagnetic
waves in dielectric medias having random or stochastic refractive indices, and propagation of ERM waves in a turbulent atmosphere or ocean with random density fluctuations
\textbf{[22-25]}. A body of work exists developed in the 1950s and 1960s on the Helmholtz equation with random parameters and boundary conditions $\mathbf{[22-25,63]}$.

Suppose a plane wave of frequency $\omega$ propagates in the z-direction in a medium with a homogeneous and constant refractive index $\mathrm{N}_{o}$. Splitting of the temporal factor $e^{-i\omega t}$ the amplitude $\mathfrak{F}(z)$ of the wave can be approximated in the parabolic or 'quasi-optic' approximation by \textbf{[63]}
\begin{align}
{\frac{d\mathfrak{F}(z)}{dz}}=i\frac{\omega}{c}\mathrm{N}_{o}{\mathfrak{F}}(z)
\end{align}
Suppose now the refractive index is a random Gaussian function $\mathbf{N}(z)$ given by
\begin{align}
{\mathbf{N}}(z)=\mathrm{N}_{o}+\alpha\mathrm{N}_{o}{\mathbf{G}}(z)=\mathrm{N}_{o}(1+\alpha\mathbf{G}(z))
\end{align}
where $\alpha>0$ is a small parameter giving the magnitude of the random fluctuations. Then the averaged index is
$\mathbb{E}\lbrace\mathlarger{\mathbf{N}}(z)\rbrace=\mathrm{N}_{o}+\alpha \mathrm{N}_{O}\mathbb{E}\lbrace\mathlarger{\mathbf{G}}(z)\rbrace=\mathrm{N}_{o}$. Equation (2.11) then becomes a stochastic linear differential equation
\begin{align}
\frac{d\widehat{\bm{\mathfrak{F}}(z)}}{dz}=i\frac{\omega}{c}\mathrm{N}(z)\widehat{\mathfrak{F}(z)}
=i\frac{\omega}{c}\mathrm{N}_{o}\widehat{{\mathfrak{F}}(z)}+i\frac{\omega}{c}\alpha\mathrm{N}_{o}\widehat{\mathfrak{F}(z)}{\mathbf{G}}(z)
\end{align}
or
\begin{align}
d\widehat{\mathfrak{F}(z)}=i\frac{\omega}{c}\mathbf{N}(z)\widehat{\mathfrak{F}(z)}dz=i\frac{\omega}{c}\mathrm{N}_{o}
\widehat{\mathfrak{F}(z)}dz+i\frac{\omega}{c}\alpha\mathrm{N}_{o}\widehat{\mathfrak{F}(z)}{\mathbf{G}}(z)dz
\end{align}
and the expectation recovers the deterministic ODE (2.11).
\begin{align}
\mathbb{E}\left\lbrace\frac{d\widehat{\mathfrak{F}(z)}}{dz}\right\rbrace=i\frac{\omega}{c}\mathrm{N}_{o}{\mathfrak{F}(z)}
+i\frac{\omega}{c}\alpha\mathrm{N}_{o}\widehat{\mathfrak{F}(z)}\mathbb{E}\big\lbrace{\mathbf{G}}(z)\big\rbrace=i\frac{\omega}{c}\mathrm{N}_{o}{\mathfrak{F}(z)}
\end{align}
Since this is a linear approximation and linear ODE, no new terms are induced upon averaging.

Returning to radiative transport and the RTE for neutrons applied to nuclear engineering, reactor noise refers to random fluctuations in the output signal from neutron detectors or sensors placed within the reactor or on its boundary. This has been an important area of investigation from the beginning since neutron flux estimates are required to determine safety and shut-down protocols. Thermal dependence of material properties and cross-sections, mechanical vibrations in the reactor and fuel rods, pressure fluctuations, fluctuations in coolant flow etc.,\textbf{ [55-59]} can all affect the neutron flux within the reactor, absorption/capture cross-sections and bulk absorption or capture coefficient $\mathsf{A}$, by making them randomly fluctuating functions in space and/or time.

For the problem of laser propagation in a medium by a randomly fluctuating absorption coefficient, a similar analysis can be applied.
\begin{prop}$\textbf{Stochastic absorption coeffient}$\newline
A stochastic or randomly fluctuating absorption coefficient can be defined as
\begin{align}
{\mathbf{A}}(z)=\underbrace{\mathsf{A}}_{mean~value}+\underbrace{\alpha\mathsf{A}\mathbf{G}(z)}_{random/fluctuating~term}\equiv \mathsf{A}+\mathbf{R}(z)
\end{align}
with mean value
\begin{align}
\mathbb{E}\big\lbrace\mathbf{A}(z)\big\rbrace=\mathsf{A}+\alpha\mathsf{A}\mathbb{E}\big\lbrace\mathbf{G}(z)\big\rbrace
\equiv\mathsf{A}+\mathbb{E}\big\lbrace{\mathbf{G}}(z)\big\rbrace=\mathsf{A}
\end{align}
since $\mathbb{E}\lbrace\mathbf{G}(z)\rbrace=0$. The binary correlation of the stochastic absorption coefficient for any pair $(z,z^{\prime})$ is
\begin{align}
\mathbb{E}\big\lbrace{\mathbf{A}}(z)\otimes{\mathbf{A}}(z^{\prime})\big\rbrace
=\mathsf{A}^{2}+\alpha^{2}\mathsf{A}^{2}\phi(z,z^{\prime},\xi)
\end{align}
Here, $\alpha>0$ determines the magnitude of the random fluctuations about the mean value $\mathsf{A}$. The random function $\mathbf{A}(z)$ is homogenous and isotropic and stationary in that for any $d>0$, and at least $\kappa\ge 2$
\begin{align}
&\mathbb{E}\big\lbrace\mathbf{A}(z+d)\otimes\mathbf{A}(z^{\prime}+d)\big\rbrace=\mathbb{E}\left\lbrace\mathbf{A}(z)\otimes\mathbf{A}(z^{\prime})\right\rbrace\nonumber\\&
=\mathsf{A}^{2}+\alpha^{2}\mathsf{A}^{2}\phi(z+d,z^{\prime}+d,\xi)=\mathsf{A}^{2}+\alpha^{2}\mathsf{A}^{2}\phi(z,z^{\prime},\xi)
\end{align}
since $\phi(z+d,z^{\prime}+d;\xi)=\phi(z,z^{\prime};\xi)$. Also
\begin{align}
&\frac{d}{dz}\mathbb{E}\left\lbrace\mathbf{A}(z)\otimes\mathbf{A}(z)\right\rbrace=2\mathbb{E}
\left\lbrace\mathbf{A}(z)\otimes\frac{d}{dz}\mathbf{A}(z)\right\rbrace\nonumber\\&
=\lim_{z\rightarrow z^{\prime}}\frac{d}{dz}\bm{\mathsf{C}}\exp(-|z-z^{\prime}|^{\kappa}\xi^{-\kappa})
=-\lim_{z\rightarrow z^{\prime}}\kappa|z-z^{\prime}|^{\kappa-1}\bm{\mathsf{C}}\exp(-|z-z^{\prime}|^{\kappa}\xi^{-\kappa})=0
\end{align}
\end{prop}
\begin{figure}[htb]
\begin{center}
\includegraphics[height=2.0in,width=6.5in]{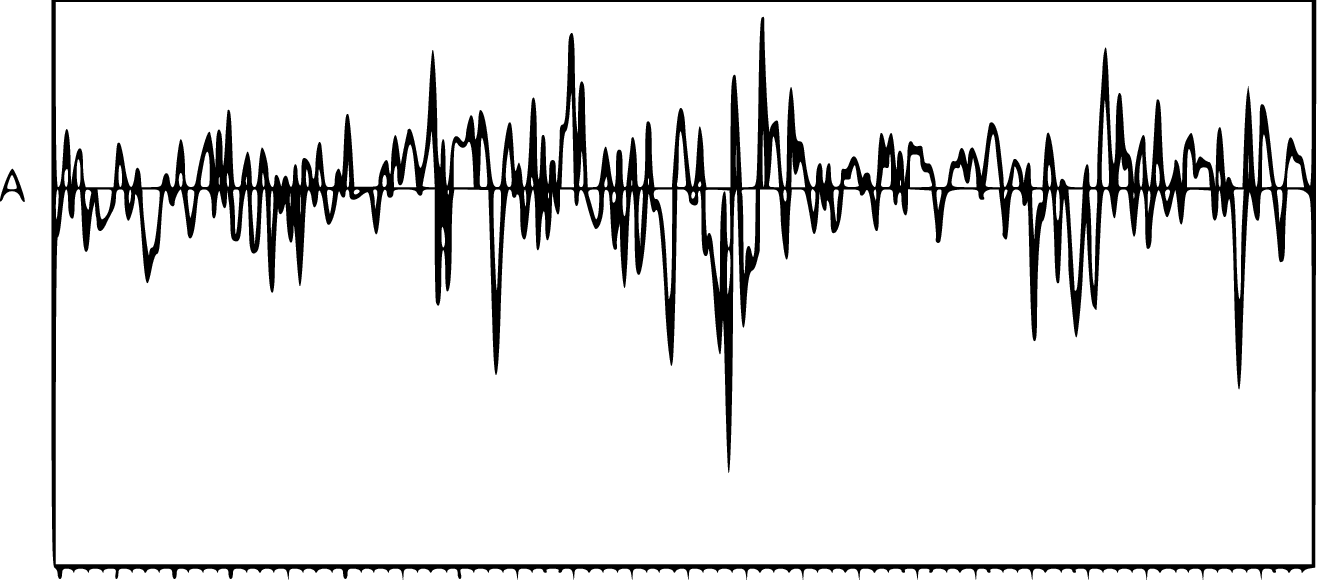}
\caption{Random Gaussian fluctuations of the absorption coefficient about the mean value $\mathsf{A}$}
\end{center}
\end{figure}
\begin{lem}
The mean free path (MFP) to absorption is $\mathsf{M}=\mathsf{A}^{-1}$ for a constant coefficient $\mathsf{A}$. Given the random function or field
$\mathbf{A}({z})$ then the stochastic MFP is
\begin{align}
\mathbf{M}({z})=|\mathbf{A}({z})|^{-1}
\end{align}
and the expectation is
\begin{align}
\mathbb{E}\big\lbrace\mathsf{M}({z})\big\rbrace=\mathbb{E}\big\lbrace|\mathbf{A}({z})|^{-1}\big\rbrace=\mathsf{A}(1+|\mathcal{S}|)=\mathsf{A}\mathcal{B}
\end{align}
where $\mathcal{S}|\ll 1$, so that the mfp is shifted or 'boosted' by a small factor $\mathcal{B}=(1+\mathcal{S})$.
\end{lem}
\begin{proof}
\begin{align}
&\mathsf{M}(z)=|\mathscr{A}({z})|^{-1}=[\mathsf{A}+\alpha\mathsf{A}{\mathbf{G}}(z)]^{-1}=
\mathsf{A}^{-1}[1+\alpha\mathbf{G}({z})]^{-1}\nonumber\\&
=\mathsf{A}^{-1}\bigg(\sum_{Q=0}^{\infty}\binom{-1}{Q}\alpha^{Q}|{\mathbf{G}}(z)|^{Q}\bigg)
\end{align}
using the binomial series. Their expectations or averages are then
\begin{align}
&\mathsf{M}(z)=\mathbb{E}\bigg\lbrace\mathbf{M}(z)\bigg\rbrace=\mathbb{E}\bigg\lbrace\frac{1}{\mathbf{A}(z)}\bigg\rbrace
=\mathsf{A}^{-1}\mathbb{E}\bigg\lbrace\big(1+\alpha{\mathbf{G}}(z))^{-1}\bigg\rbrace\nonumber\\&
=\mathsf{A}^{-1}\left(1+\sum_{Q=1}^{\infty}\binom{-1}{Q}|\mathcal{R}(\alpha,\bm{\mathsf{C}},Q|^{Q}\right)
=\mathsf{A}^{-1}(1+\mathcal{S})
\end{align}
where $\mathcal{S}=\sum_{Q=1}^{\infty}\binom{-1}{Q}|\mathcal{R}(\alpha,\bm{\mathsf{C}},Q|^{Q} $ and
\begin{align}
&\mathcal{R}(\alpha,\bm{\mathsf{C}},Q|=
\alpha\left\lbrace\frac{1}{2}(\bm{\mathsf{C}}^{Q/2}+(-1)^{Q}\bm{\mathsf{C}}^{Q/2})\right\rbrace^{1/Q}=\alpha\mathbb{E}\lbrace |\mathbf{G}(z)|^{Q}\rbrace
\end{align}
The (binomial) series will converge if $\mathcal{R}(\beta,\bm{\mathsf{C}},Q)<1$. On averaging, the mfp is then shifted by a
small factor $\mathcal{B}=(1+\mathcal{S})$, and $\mathcal{S}$ is small if $\alpha$ and $\bm{\mathsf{C}}$ are small.
\end{proof}
\section{\textbf{Stochastic differential equation for beam propagation in stochastic absorbing medias}}
A stochastic differential equation  for laser propagation and attenuation within matter with a 'noisy' or random Gaussian absorption coefficient $\mathbf{A}(\mathbf{x})$ can then be formulated as follows.
\begin{prop}(\textbf{Stochastic differential equation for beam propagation})\newline
Let absorbing matter have a random Gaussian absorption coefficient $\mathbf{A}(z)=\mathsf{A}+\alpha\mathsf{A}\mathbf{G}(z)$
as previously defined. Replacing $\mathsf{A}$ with $\mathbf{A}(z)$ gives
\begin{align}
&\frac{\widehat{\bm{\psi}(z,\widehat{\mathbf{e}}_{3})}}{d z}=\mathbf{A}(z)\widehat{\bm{\psi}(z,\widehat{\mathbf{e}}_{3})}\nonumber\\&
=-(\mathsf{A}+\alpha\mathsf{A}{\mathbf{G}}(z))\widehat{\bm{\psi}(z,\widehat{\mathbf{e}}_{3})},~~~~~~z\in\mathbb{I\!D}=[0,\mathrm{L}]
\end{align}
In differential form
\begin{align}
&d\widehat{\bm{\psi}(z,\widehat{\mathbf{e}}_{3})}=-\mathbf{A}(z))\widehat{\bm{\psi}(z,\widehat{\mathbf{e}}_{3})}dz\nonumber\\&
=-(\mathsf{A}+\alpha {A}\mathbf{G}(z))\widehat{\bm{\psi}(z,\widehat{\mathbf{e}}_{3})}dz
=-\mathsf{A}\widehat{\bm{\psi}(z,\widehat{\mathbf{e}}_{3})}dz-\alpha\mathsf{A}\widehat{\bm{\psi}(z,\widehat{\mathbf{e}}_{3})}\mathbf{G}(z)dz,~~~~~~z\in\mathbb{I\!D}
\end{align}
Since the flux/intensity solution itself will now be stochastic, we write $\widehat{\bm{\psi}(z,\widehat{\mathbf{e}}_{3})}$ instead of
${\psi}(z,\widehat{\mathbf{e}}_{3})$. Then
\begin{align}
&\frac{d\widehat{\bm{\psi}(z,\widehat{\mathbf{e}}_{3}})}{\widehat{\bm{\psi}(z,\widehat{\mathbf{e}}_{3})}}=-\mathbf{A}(z))dz
=-(\mathsf{A}+\alpha\mathsf{A}{\mathbf{G}}(z))dz\nonumber
\\&
=-\mathsf{A}dz-\alpha\mathsf{A}{\mathbf{G}}(z)dz,~~~~~~z\in\mathbb{I\!D}
\end{align}
\end{prop}
The linear SDE (3.2) can be solved exactly.
\begin{thm}
The exact solution of this SDE is given by
\begin{align}
&\widehat{\bm{\psi}(z,\widehat{\mathbf{e}}_{3})}=\psi_{o}\exp(-\mathsf{A}~z)\exp\left(-\alpha\mathsf{A}\int_{0}^{z}\mathbf{G}(\bar{z})d\bar{z}
\right)\nonumber\\&=\widehat{\bm{\psi}(z,\widehat{\mathbf{e}}_{3})}\exp\left(-\alpha\mathsf{A}\int_{0}^{z}\mathbf{G}(\bar{z})d\bar{z}
\right)
\end{align}
with stochastic expectation
\begin{align}
&\mathbb{I}(z,\widehat{\mathbf{e}}_{3})=\mathbb{E}\big\lbrace\widehat{\bm{\psi}(z,\widehat{\mathbf{e}}_{3})}\big\rbrace=\psi_{o}\exp(-\mathsf{A}~z)\mathbb{E}\left\lbrace
\exp\left(-\alpha\mathsf{A}\int_{0}^{z}\mathbf{G}(\hat{z})d\hat{z}
\right)\right\rbrace\nonumber\\&=\psi(z,\widehat{\mathbf{e}}_{3})\mathbb{E}\left\lbrace\exp\left(-\alpha\mathsf{A}\int_{0}^{z}\mathbf{G}(\hat{z})d\hat{z}
\right)\right\rbrace
\end{align}
\end{thm}
\begin{proof}
Integrating (3.3)
\begin{align}
&\int_{\psi_{o}}^{\widehat{\bm{\psi}(z,\widehat{\mathbf{e}}_{3})}}
\frac{d\widehat{\bm{\Psi}(z,\widehat{\mathbf{e}}_{3})}}{\widehat{\bm{\Psi}(z,\widehat{\mathbf{e}}_{3})}}=\int_{0}^{z}\mathbf{A}(z))dz
=-\int_{0}^{z}(\mathsf{A}+\alpha\mathsf{A}\mathbf{G}(z))dz\nonumber
\\&=-\int_{0}^{z}\mathsf{A}dz-\alpha\mathsf{A}\int_{0}^{z}\mathbf{G}(z)dz
,~~~~~~z\in\mathbb{I\!D}
\end{align}
Then
\begin{align}
\log\widehat{\bm{\psi}(z,\widehat{\mathbf{e}}_{3})}-\log \psi_{o}\equiv
\log\left(\frac{\widehat{\bm{\psi}(z,\widehat{\mathbf{e}}_{3})}}{\psi_{o}}\right)
=-\int_{0}^{z}\mathsf{A}d\bar{z}-\alpha\mathsf{A}\int_{0}^{z}\mathbf{G}(\bar{z})d\bar{z}
\end{align}
Equation (3.4) then follows from taking the exponential of both sides, and (3.5) from taking the expectation.
\end{proof}
Note that since the Gaussian function $\mathbf{G}(z)$ is non-white with a regulated binary correlation, then its derivative exists and can be defined and so regular calculus can be used to find the solution.

The expectation (3.5) can be evaluated via a Van Kampen-type cumulant or cluster expansion \textbf{[60-63]}. This is a useful tool
in evaluating expectations of exponentials of stochastic integrals. Due to $\mathbf{G}(z)$ being Gaussian, all high-order cumulants beyond 2nd order are equal to zero.
\begin{thm}(\textbf{Cumulant expansion of the stochastic exponential})\newline
The expectation of the solution has the following Van-Kampen-type cumulant expansion
\begin{align}
&\mathbb{I}(z,\widehat{\mathbf{e}}_{3})=\bigg\lbrace\widehat{\bm{\psi}(z,\widehat{\mathbb{e}}_{3})}\bigg\rbrace=\psi_{o}\exp(-\mathsf{A}~z)
\mathbb{E}\left\lbrace\exp\left(-\alpha\mathsf{A}\int_{0}^{z}{\mathbf{G}}(z)dz\right)\right\rbrace\nonumber\\&
=\psi_{o}\exp(-\mathsf{A}~z)\exp\left(\sum_{m=1}^{\infty}\frac{-(\alpha\mathsf{A})^{m}}{m!}\mathlarger{\int}_{\mathbb{D}}
\bm{\mathcal{D}}_{m}[z]{\mathbb{K}}\left\lbrace\prod_{q=0}^{m}\otimes{\mathbf{G}}(z)\right\rbrace\right)\nonumber\\&
\equiv\psi_{o}\exp(-\mathsf{A}~z)\exp\left(\sum_{m=1}^{\infty}\frac{-(\alpha\mathsf{A})^{m}}{m!}\mathlarger{\int}_{0}^{z}...\mathlarger{\int}_{0}^{z_{m-1}}dz_{1}...dz_{m}
\mathbb{K}\bigg\lbrace{\mathbf{G}}(z_{1})\otimes...\otimes{\mathbf{G}}(z_{m})\bigg\rbrace\right)
\end{align}
\end{thm}
\begin{proof}
Given a set of planes $(z_{1},...z_{m})\in\mathbb{I\!D}\subset\mathbb{R}^{(+)}$, the mth-order moments and cumulants for the GRFs are given by (2.7) and (2.8) so that
\begin{align}
&\overrightarrow{\mathbf{N}}\mathbb{E}\bigg\lbrace{\mathbf{G}}(z_{1})\otimes...\otimes{\mathbf{G}}(z_{m})\bigg\rbrace
=\overrightarrow{\mathbf{N}}\mathbb{E}\left\lbrace\prod_{q=1}^{m}{\mathbf{G}}(z_{q})\right\rbrace\\&
\overrightarrow{\mathbf{N}}\mathbb{K}\bigg\lbrace{\mathbf{G}}(z_{1})\otimes...\otimes{\mathbf{G}}(z_{m})\bigg\rbrace
=\overrightarrow{\mathbf{N}}\mathbb{K}\left\lbrace\prod_{q=1}^{m}{\mathbf{G}}(z_{q})\right\rbrace\\&
\overrightarrow{\mathbf{N}}\mathbb{K}\bigg\lbrace{\mathbf{G}}(z_{1})\otimes...\otimes{\mathbf{G}}(z_{m})\bigg\rbrace
=\overrightarrow{\mathbf{N}}\mathbb{K}\left\lbrace\prod_{q=1}^{m}{\mathbf{G}}(z_{q})\right\rbrace_{m\ge 3}=0
\end{align}
where $\overrightarrow{\mathbf{N}}$ is a 'normal ordering operator'. For example $\overrightarrow{\mathbf{N}} f(t_{2})f(t_{3})f(t_{1})=f(t_{1})f(t_{2})f(t_{3})$ for $t_{1}<t_{2}<t_{3}$ and so on, for any function $f(t)$.

The moment generating function (MGF) $\mathlarger{\mathscr{M}}[\mathbf{G}(z)]$ and the cumulant-generating function [CGF]
$\mathlarger{\mathscr{C}}[\mathbf{G}(z)]$ are given by
\begin{align}
&\mathlarger{\mathscr{M}}[\mathbf{G}(z)]=\sum_{m=0}^{\infty}\frac{\beta^{m}}{m!}\int_{0}^{z}...\int_{0}^{z_{m-1}}dz_{1}...dz_{m}\overrightarrow{\mathbf{N}}
\mathbb{E}\left\lbrace\prod_{q=1}^{m}{\mathbf{G}}(z_{q})
\right\rbrace\\&
\mathlarger{\mathscr{C}}[{\mathbf{G}}(z)]=\sum_{m=1}^{\infty}\frac{\beta^{m}}{m!}\int_{0}^{z}...\int_{0}^{z_{m-1}}dz_{1}...dz_{m}
\overrightarrow{\mathbf{N}}\mathbb{K}\left\lbrace\prod_{q=1}^{m}{\mathbf{G}}(z_{q})
\right\rbrace
\end{align}
where $\beta$ is an arbitrary constant. It is important to note that the summation in (3.13) begins from $m=1$ and not $m=0$. These can also be written as
\begin{align}
&\mathlarger{\mathscr{M}}[\mathbf{G}(z)]=\sum_{m=0}^{\infty}\frac{\mathlarger{\beta}^{m}}{m!}\int\bm{\mathcal{D}}_{m}[z]\overrightarrow{\mathbf{N}}
\mathbb{E}\left\lbrace\prod_{q=1}^{m}{\mathbf{G}}(z_{q})
\right\rbrace\\&
\mathlarger{\mathscr{C}}[\mathbf{G}(z)]=\sum_{m=1}^{\infty}\frac{\mathlarger{\beta}^{m}}{m!}\int\bm{\mathcal{D}}_{m}[z]
\overrightarrow{\mathbf{N}}\mathbb{K}\left\lbrace\prod_{q=1}^{m}{\mathbf{G}}(z_{q})
\right\rbrace
\end{align}
where $\int\bm{\mathcal{D}}_{m}(z)=\int...\int dz_{1}...dz_{m}$ is a 'path integral'. Choosing ${\beta}=-\alpha\mathsf{A}$
\begin{align}
&\mathlarger{\mathscr{M}}[\mathbf{G}(z)]=\sum_{m=0}^{\infty}\frac{(-1)^{m}(\alpha\mathsf{A})^{m}}{m!}\int\bm{\mathcal{D}}_{m}[z]\overrightarrow{\mathbf{N}}\mathbb{E}\left\lbrace\prod_{q=1}^{m}
{\mathbf{G}}(z_{q})\right\rbrace\\&
\mathlarger{\mathscr{C}}[\mathbf{G}(z)]=\sum_{m=1}^{\infty}\frac{(-1)^{m}(\alpha\mathsf{A})^{m}}{m!}\int\bm{\mathcal{D}}_{m}[z]\overrightarrow{\mathbf{N}}\mathbb{K}\left\lbrace\prod_{q=1}^{m}
{\mathbf{G}}(z_{q})
\right\rbrace
\end{align}
The relation between the MGF and the CGF is \textbf{[60,62]}
\begin{align}
\log\mathlarger{\mathscr{M}}[\mathbf{G}(z)]=\mathlarger{\mathscr{C}}[\mathbf{G}(z)]
\end{align}
so that
\begin{align}
\mathlarger{\mathscr{M}}[\mathbf{G}(z)]=\exp\big(\mathlarger{\mathscr{C}}[\mathbf{G}(z)]\big)
\end{align}
Hence
\begin{align}
&\sum_{m=0}^{\infty}\frac{(-1)^{m}(\alpha\mathsf{A})^{m}}{m!}\int\bm{\mathcal{D}}_{m}[z]\overrightarrow{\mathbf{N}}\mathbb{E}\left\lbrace\prod_{q=1}^{m}{\mathbf{G}}(z_{q})
\right\rbrace\\&
=\exp\left(\sum_{m=1}^{\infty}\frac{(-1)^{m}(\alpha\mathsf{A})^{m}}{m!}\int\bm{\mathcal{D}}_{m}[z]\overrightarrow{\mathbf{N}}\mathbb{K}\left\lbrace
\prod_{q=1}^{m}\mathbf{G}(z_{q})
\right\rbrace\right)
\end{align}
However
\begin{align}
&\mathlarger{\mathscr{M}}[\mathbf{G}(z)]\equiv\mathbb{E}\left\lbrace\exp\left(-\alpha\mathsf{A}\int_{0}^{z}{\mathbf{G}}(\bar{z})d\bar{z}\right)\right\rbrace\nonumber\\&
=\sum_{m=0}^{\infty}\frac{(-1)^{m}(\alpha\mathsf{A})^{m}}{m!}\int\bm{\mathcal{D}}_{m}[z]\overrightarrow{\mathbf{N}}\mathbb{E}\left\lbrace
\prod_{q=1}^{m}{\mathbf{G}}(z_{q})
\right\rbrace
\end{align}
giving a perturbation series in terms of cumulants to all orders.
\begin{align}
&\mathlarger{\mathscr{M}}[\mathbf{G}(z)]\equiv \mathbb{E}\left\lbrace\exp\left(-\alpha\mathsf{A}\int_{0}^{z}\mathbf{G}(\bar{z})d\bar{z}\right)\right\rbrace\nonumber\\&
=\sum_{m=0}^{\infty}\frac{(-1)^{m}(\alpha\mathsf{A})^{m}}{m!}\int\bm{\mathcal{D}}_{m}[z]\overrightarrow{\mathbf{N}}\mathbb{E}\left\lbrace\prod_{q=1}^{m}{\mathbf{G}}(z_{q})
\right\rbrace\nonumber\\&
=\exp\left(\sum_{m=1}^{\infty}\frac{(-1)^{m}(\alpha\mathsf{A})^{m}}{m!}\int\bm{\mathcal{D}}_{m}[z]\overrightarrow{\mathbf{N}}\mathbb{K}\left\lbrace\prod_{q=1}^{m}{\mathbf{G}}(z_{q})
\right\rbrace\right)
\end{align}
The solution is then (3.8).
\end{proof}
With Gaussian random functions, the summation only needs to be taken to 2nd order (m=2) since all higher-order cumulants are zero. That is
\begin{align}
\sum_{m=3}^{\infty}\frac{(-1)^{m}(\alpha\mathsf{A})^{m}}{m!}\int\bm{\mathcal{D}}_{m}[z]\overrightarrow{\mathbf{N}}\mathbb{K}\left\lbrace\prod_{q=1}^{m}{\mathbf{G}}(z_{q})
\right\rbrace=0
\end{align}
\begin{lem}
\begin{align}
&\mathbb{I}(z,\widehat{\mathbf{e}}_{3})=\mathbb{E}\big\lbrace\widehat{\bm{\psi}(z,\widehat{\mathbf{e}}_{3})}\big\rbrace=\psi_{o}
\exp(-\mathsf{A}~z)\mathbb{E}\left\lbrace
\exp\left(-\alpha\mathsf{A}\int_{0}^{z}{\mathbf{G}}(z)dz\right)\right\rbrace\nonumber\\&=\psi_{o}\exp(-\mathsf{A}~z)
\exp\left(\alpha^{2}\mathsf{A}^{2}\int_{0}^{z}\int_{0}^{z_{2}}dz_{1}dz_{2}\phi(z_{1},z_{2};\xi)\right)
\end{align}
\end{lem}
\begin{proof}
Expanding (3.8) to second order, all high-order terms vanish
\begin{align}
&\mathbb{I}(z,\widehat{\mathbf{e}}_{3})=\big\lbrace\widehat{\bm{\psi}(z,\widehat{\mathbf{e}}_{3})}\bigg\rbrace=\psi_{o}\exp(-\mathsf{A}~z)\mathbb{E}\left\lbrace
\exp\left(-\alpha\mathsf{A}\int_{0}^{z}{\mathbf{G}}(z)dz\right)\right\rbrace\nonumber\\&
=\psi_{o}\exp(-\mathsf{A}~z)\exp\left(\sum_{m=1}^{\infty}\frac{(-1)^{m}(\alpha\mathsf{A})^{m}}{m!}\int\bm{\mathcal{D}}_{m}[z]\overrightarrow{\mathbf{N}}\mathbb{K}\left\lbrace\prod_{q=1}^{m}
\mathbf{G}(z_{q})\right\rbrace\right)\nonumber\\&\nonumber\\&
=\psi_{o}\exp(-\mathsf{A}~z)\exp\left(\sum_{m=1}^{2}\frac{(-1)^{m}(\alpha\mathsf{A})^{m}}{m!}\int\bm{\mathcal{D}}_{m}[z]\overrightarrow{\mathbf{N}}\mathbb{K}\left\lbrace\prod_{q=1}^{m}
\mathbf{G}(z_{q})\right\rbrace\right)\nonumber\\&
=\psi_{o}\exp(-\mathsf{A}~z)\exp\bigg(-\alpha\mathsf{A}\int_{0}^{z}dz_{1}\mathbb{K}\big\lbrace{\mathbf{G}}(z)\big\rbrace\nonumber\\&
+\alpha^{2}\mathsf{A}^{2}\int_{0}^{z}\int_{0}^{z_{1}}dz_{1}dz_{2}\mathbb{K}\big\lbrace{\mathbf{G}}(z_{1})\otimes{\mathbf{G}}(z_{2})
\big\rbrace\bigg)\nonumber\\&
=\psi_{o}\exp(-\mathsf{A}~z)\exp\bigg(-\alpha\mathsf{A}\int_{0}^{z}dz_{1}\mathbf{E}\big\lbrace{\mathbf{G}}(z)\big\rbrace\nonumber\\&
+\alpha^{2}\mathsf{A}^{2}\int_{0}^{z}\int_{0}^{z_{1}}dz_{1}dz_{2}\mathbb{E}\big\lbrace{\mathbf{G}}(z_{1}){\otimes}{\mathbf{G}}(z_{2})
\big\rbrace\bigg)\nonumber\\&
=\psi_{o}\exp(-\mathsf{A}~z)\exp\left(\alpha^{2}\mathsf{A}^{2}\int_{0}^{z}\int_{0}^{z_{1}}dz_{1}dz_{2}\mathbb{E}\big\lbrace{\mathbf{Q}}(z_{1})
\otimes{\mathbf{Q}}(z_{2})
\big\rbrace\right)\nonumber\\&
=\psi_{o}\exp(-\mathsf{A}~z)\exp\left(\alpha^{2}\mathsf{A}^{2}\int_{0}^{z}\int_{0}^{z_{1}}dz_{1}dz_{2}
\phi(z_{1},z_{2};\xi)\right)
\end{align}
since $\mathbb{K}\lbrace\mathbf{G}(z_{1})\rbrace\equiv \mathbb{E}\lbrace\mathbf{G}(z_{1})\rbrace=0$ and $\mathbb{K}\lbrace\mathbf{G}(z_{1})\otimes\mathbf{G}(z_{2})\rbrace
=\mathbb{E}\lbrace\mathbf{G}(z_{1})\otimes\mathbf{G}(z_{2})\rbrace=\phi(z_{1},z_{2};\xi)$
\end{proof}
\begin{cor}
\begin{align}
\widehat{\bm{\psi}(z,\widehat{\mathbf{e}}_{3})}\longrightarrow\psi_{o}\exp(-\mathsf{A}~z)
\end{align}
as $\alpha\rightarrow 0$ or $\mathbf{G}(z)\rightarrow 0$ and the original Beer's law solution is recovered.
\end{cor}
\begin{lem}
The stochastic average $\mathbb{I}(z,\widehat{\mathbf{e}}_{3})$ is a solution of the ordinary differential equation
\begin{align}
\frac{d\mathbb{I}(z,\widehat{\mathbf{e}}_{3})}{dz}=\mathsf{A}\big(\alpha^{2}\mathsf{A}\theta(z)-1\big)\mathbb{I}(z,\widehat{\mathbf{e}}_{3})
\end{align}
where $\theta(z)$ is the integral
\begin{align}
\theta(z_{1})=\int_{0}^{z_{1}}\phi(z_{1},z_{2};\xi)dz_{2}
\end{align}
and so Beer's Law is recovered for $\alpha=0$.
\end{lem}
\begin{proof}
Taking the derivative of (3.25) gives
\begin{align}
&\frac{d\mathbb{I}(z,\widehat{\mathbf{e}}_{3}}{dz}=-\mathsf{A}\psi_{o}\exp(-\mathsf{A}z)\exp\left(\alpha^{2}\mathsf{A}^{2}\int_{0}^{z}\theta(z_{1}) dz_{1}\right)\nonumber\\&+\psi_{o}\exp(-\mathsf{A}z)\alpha^{2}\mathsf{A}^{2}\theta(z)\exp\left(\alpha^{2}\mathsf{A}^{2}\int_{0}^{z}\theta(z_{1}) dz_{1}\right)
\end{align}
and so (3.28) follows
\end{proof}
\subsection{Estimate of $\mathbb{E}\big\lbrace\widehat{\bm{\psi}(z,\widehat{\mathbf{e}}_{3})}\big\rbrace $ for a Gaussian decay correlation}
The expression  is now explicitely computed for a Gaussian correlation function
\begin{align}
&\mathbb{E}\big\lbrace\widehat{\bm{\psi}(z,\widehat{\mathbf{e}})}\big\rbrace\equiv\mathbb{I}(z,\widehat{\mathbf{e}}_{3})=
\psi_{o}\exp(-\mathsf{A}z)\exp\left(\frac{1}{2}\alpha^{2}\mathsf{A}^{2}
\bm{\mathsf{C}}\int_{0}^{z}\int_{0}^{z_{1}}\phi(z_{1},z_{2};\xi)dz_{1}dz_{2}\right)\nonumber\\&=\psi_{o}\exp(-\mathsf{A}z)
\exp\left(\frac{1}{2}\alpha^{2}\mathsf{A}^{2}
\bm{\mathsf{C}}\int_{0}^{z}\int_{0}^{z_{1}}\exp(-|z_{1}-z_{2}|^{2}\xi^{-2}dz_{1}dz_{2}\right)\nonumber\\&
=\psi_{o}\exp(-\mathsf{A}z)\exp\left(\frac{1}{2}\alpha^{2}\mathsf{A}^{2}
\bm{\mathsf{C}}{\mathcal{Y}}(z)\right)=\psi_{o}\exp(-\mathsf{A}z){\mathcal{B}}(\alpha,\mathsf{A},z)
\end{align}
This is then Beer's law of exponential attenuation $\psi_{o}\exp(-\mathsf{A} z)$ multiplied by a small 'boost' factor $\mathcal{B}(\alpha,\mathsf{A},z)$.
Note that because $\alpha$ is small for small fluctuations about the mean with $\alpha\sim 0$ or $\alpha\ll 1$, the exponential term $\exp\left(\frac{1}{2}\alpha^{2}\mathsf{A}^{2}\bm{\mathsf{C}}{\mathcal{Y}}(z)\right)$ will be close to one. It just remains to
compute the double integral over the Gaussian, namely $\mathcal{Y}(z)$.
\begin{lem}
\begin{align}
{\mathcal{Y}}(z)=\int_{0}^{z}\int_{0}^{z_{1}}\exp(-|z_{1}-z_{2}|^{2}\xi^{-2})dz_{1}dz_{2}
\end{align}
then
\begin{align}
{\mathcal{Y}}(z)=\int_{0}^{z}\bigg|\int_{0}^{z_{1}}\exp(-|z_{1}-z_{2}|^{2}\xi^{-2})dz_{2}\bigg|dz_{1}=\int_{0}^{z}{\mathcal{W}}(z_{1})dz_{1}
\end{align}
The evaluated integral is then
\begin{align}
{\mathcal{Y}}(z)=\int_{0}^{z}\mathcal{W}(z_{1})dz_{1}=\frac{1}{2}\xi\left[\exp(-z^{2}/\xi^{2})\left(\sqrt{\pi}z Erf\left(\frac{z}{\xi}\right)\exp\left(\frac{z^{2}}{\xi^{2}}\right)+{\xi}\right)-{\xi}\right]
\end{align}
\end{lem}
\begin{proof}
The inner integral is an incomplete Gaussian integral which gives the standard error function Erf so that
\begin{align}
\mathcal{W}(z_{1})=\int_{0}^{z_{1}}\exp(-|z_{1}-z_{2}|^{2}\xi^{-2})dz_{2}=\frac{1}{2}\sqrt{\pi}Erf\left(\frac{z_{1}}{\xi}\right)\xi
\end{align}
The error function has the properties $Erf(0)=0$ and $Erf(\infty)=1$. Then
\begin{align}
{\mathcal{Y}}(z)=\int_{0}^{z}{\mathcal{W}}(z_{1})dz_{1}=\frac{1}{2}\sqrt{\pi}\int_{0}^{z}Erf\left(\frac{z_{1}}{\xi}\right)\xi dz_{1}
\end{align}
The full integral is then
\begin{align}
\mathcal{Y}(z)=\int_{0}^{z}\mathcal{W}(z_{1})dz_{1}= \frac{1}{2}\xi\left[\exp(-z^{2}/\xi^{2})\left(\sqrt{\pi}z Erf\left(\frac{z}{\xi}\right)\exp\left(\frac{z^{2}}{\xi^{2}}\right)+\xi\right)-\xi\right]
\end{align}
and so
\begin{align}
&\exp\left(\frac{1}{2}\alpha^{2}\mathsf{A}^{2}\bm{\mathsf{C}}{\mathcal{Y}}(z)\right)
=\exp\left(\frac{1}{2}\alpha^{2}\mathsf{A}^{2}\bm{\mathsf{C}}\int_{0}^{z}\bigg|\int_{0}^{z_{1}}\exp(-|z_{1}-z_{2}|^{2}\xi^{-2})dz_{2}\bigg|dz_{1}\right)\nonumber\\&
=\exp\left(\frac{1}{4}\alpha^{2}\mathsf{A}^{2}\bm{\mathsf{C}}\xi\left[\exp(-z^{2}/\xi^{2})\left(\sqrt{\pi}z Erf\left(\frac{z}{\xi}\right)\exp\left(\frac{z^{2}}{\xi^{2}}\right)+\xi\right)-\xi\right]\right)
\end{align}
\end{proof}
Note that $\mathlarger{\mathcal{Y}}(z)$ is dimensionless as required. The following theorem has now been established and proved.
\begin{thm}
Let $\mathbb{I\!D}\subset\mathbb{R}^{3}$ be a slab geometry along of depth $\mathrm{L}$ with boundaries $z=0$ and $z=\mathrm{L}$ and $\mathrm{L}\gg 1$. A monochromatic laser with a flat profile and incident intensity $\psi_{o}$ is incident upon the slab along the z-axis or unit vector $\widehat{\mathbf{e}}_{3}$ at $z=0$. The slab contains purely absorbing matter with an absorption coefficient of $\mathsf{A}$ with respect to the wavelength of the laser light. If $\mathsf{A}$ is constant and homogenous then the beam decays as Beer's law
$ \psi(z,\widehat{\mathbf{e}}_{3})=\psi_{o}\exp(-\mathsf{A} z)$. If the absorption coefficient is randomly fluctuating in space as
\begin{align}
\mathbf{A}(z)=\mathsf{A}+\alpha \mathsf{A}{\mathbf{G}}(z)
\end{align}
where $\alpha>0$ is a small parameter determining the magnitude of the fluctuations, and the Gaussian random function or noise has the binary correlation
\begin{align}
\mathbb{E}\big\lbrace {\mathbf{G}}(z_{1}){\otimes}{\mathbf{G}}(z_{2})\big\rbrace=
\phi(z_{1},z_{2};\xi)=\bm{\mathsf{C}}\exp(-|z_{1}-z_{2}|^{2}\xi^{-2})
\end{align}
with $\mathbb{E}\lbrace\mathbf{G}(z)\rbrace=0$ and correlation length $\xi$. The beam propagation and attenuation within the stochastic medium is then described by the stochastic differential equation
\begin{align}
d\widehat{\bm{\psi}(z,\widehat{\mathbf{e}}_{3})}=-\mathsf{A}\widehat{\bm{\psi}(z,\mathbf{e}_{3})}dz-\alpha\mathsf{A}\widehat{\bm{\psi}(z,\mathbf{e}_{3})}
{\mathbf{G}}(z)dz
\end{align}
The stochastic average of the solution of the SDE is then a modified Beer's law of the form
\begin{align}
&\mathbb{I}(z,\widehat{\mathbf{e}}_{3})=\mathbb{E}\big\lbrace\widehat{\bm{\psi}(z,\widehat{\mathbf{e}}_{3})}\big\rbrace
=\psi_{o}\exp(-\mathsf{A}z)
\exp\left(\frac{1}{2}\alpha^{2}\mathsf{A}^{2}\bm{\mathsf{C}}\int_{0}^{z}\bigg|\int_{0}^{z_{1}}\mathlarger{\phi}(z_{1},z_{2};\xi))dz_{2}\bigg|dz_{1}\right)\nonumber\\&
=\psi_{o}\exp(-\mathsf{A}z)
\exp\left(\frac{1}{2}\alpha^{2}\mathsf{A}^{2}\bm{\mathsf{C}}\int_{0}^{z}\bigg|\int_{0}^{z_{1}}\exp(-|z_{1}-z_{2}|^{2}\xi^{-2})dz_{2}\bigg|dz_{1}\right)\nonumber\\&
=\psi_{o}\exp(-\mathsf{A}z)\exp\left(\frac{1}{4}\alpha^{2}\mathsf{A}^{2}\bm{\mathsf{C}}\xi\left[\exp(-z^{2}/\xi^{2})\left(\sqrt{\pi}z Erf\left(\frac{z}{\xi}\right)\exp\left(\frac{z^{2}}{\xi^{2}}\right)+\xi\right)-\xi\right]\right)
\end{align}
\end{thm}
\subsection{Plots of $\mathbb{I}(z,\widehat{\mathbf{e}_{3}})$}
If we set $\mathsf{A}=1cm^{-1}$ and $\xi=1cm$ and also $\bm{\mathsf{C}}=1$ then
\begin{align}
&\mathbb{I}(z,\widehat{\mathbf{E}}_{3})=\mathbb{E}\big\lbrace\widehat{\bm{\psi}(z,\widehat{\mathbf{e}}_{3})}\big\rbrace\nonumber\\&
=\psi_{o}\exp(-z)\exp\bigg(\frac{1}{4}\alpha^{2}\bigg[\exp(-z^{2})\bigg(\sqrt{\pi}z Erf\bigg({z}\bigg)
\exp\bigg(z^{2}\bigg)+1\bigg)-1\bigg]\bigg)
\end{align}
This function can then be plotted for various values of $\alpha$ and compared with the standard Beer's Law which is $\psi(z)=\psi_{o}\exp(-z)$.  Also, one can set $\psi_{o}=1Wcm^{-2}$. These plots are for illustration to demonstrate the quantitative behavior of the expression. The results are shown in Figure 3. The effect of averaging over the fluctuations is that the beam is slightly less attenuated since the mean free path of absorption is now slightly increased. The magnitude of this effect depends on the value of $\alpha$ which determines the magnitude of the random fluctuation in the absorption coefficient. In reality, $\alpha$ will be very small so the deviation from Beer's Law will also be small.
\begin{figure}[htb]
\begin{center}
\includegraphics[height=6.0in,width=6.0in]{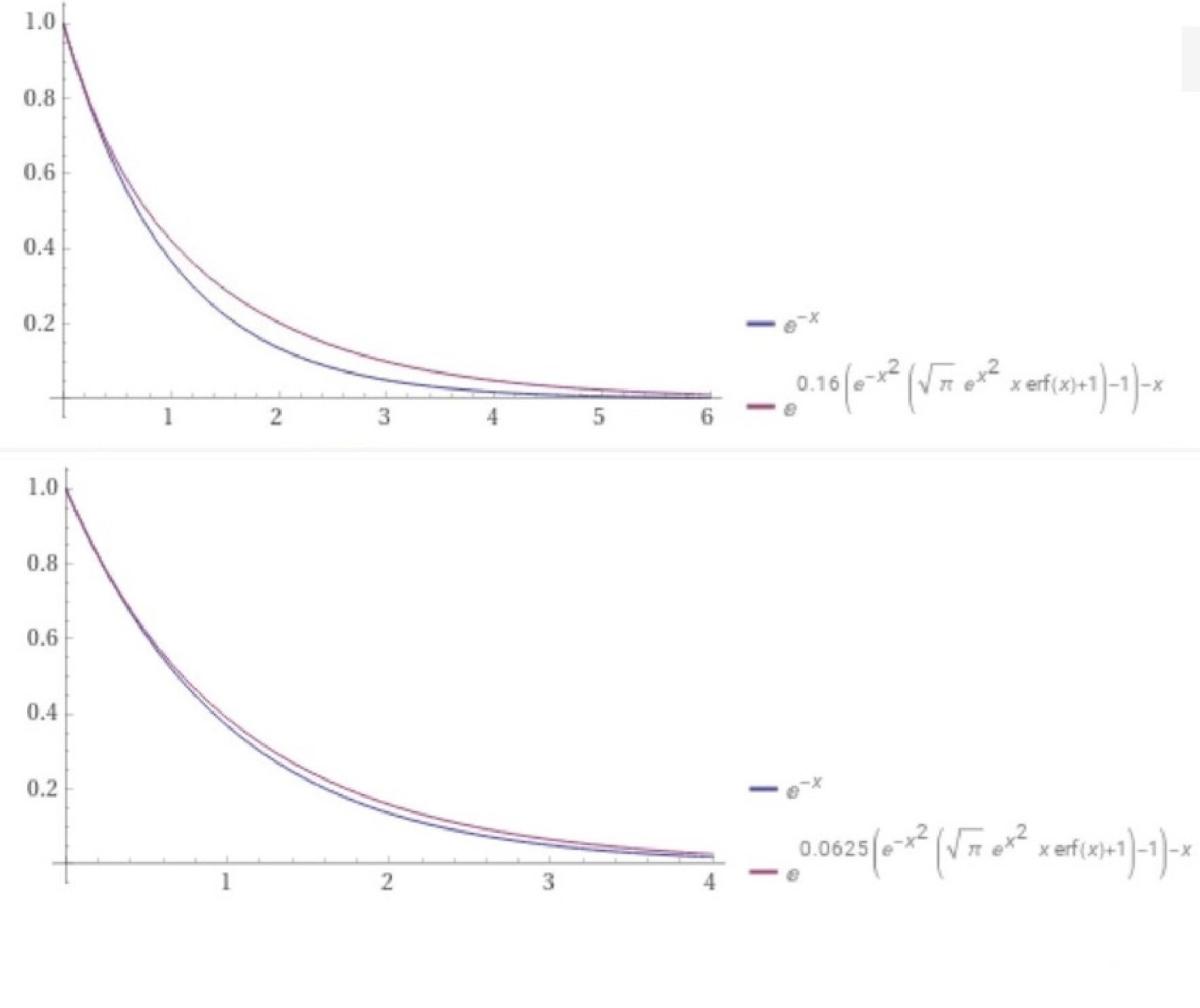}
\caption{Beam intensity penetration into the slab along the z-axis. The blue curve is the standard Beer's Law $\psi(z)$ and the purple curve is the averaged Beer's law solution $\mathbb{I}(z,\mathbf{e}_{3})=\mathbb{E}\lbrace\widehat{\psi(z,\widehat{\mathbf{e}_{3}})}\rbrace $ for a random absorption coefficient. The values of $\alpha$ are $0.8$ for the top figure and $0.5$ for the bottom figure. The blue and purple curves converge together as $\alpha\rightarrow 0$.}
\end{center}
\end{figure}
\subsection{Beer-Lambert Law}
The result automatically extends to the Beer-Lambert Law. If the scattering coefficient is non-zero then an initial beam of intensity $\psi_{o}$ will decay as
\begin{align}
\psi(z,\widehat{\mathbf{e}}_{3})=\psi_{o}\exp(-(\mathsf{A}+\mathsf{S})z)=\psi_{o}\exp(-\mathsf{T}z)
\end{align}
where $\mathsf{T}=\mathsf{A}+\mathsf{S}$ is the total attenuation coefficient. If the coefficients are noisy or
Gaussian random functions then
\begin{align}
&\mathbf{A}(z)=\mathsf{A}+\alpha\mathsf{A}\mathbf{G}(z)\\&
\mathbf{S}(z)=\mathsf{S}+\alpha\mathsf{S}\mathbf{G}(z)
\end{align}
The stochastic average of the solution of the SDE is then a modified Beer-Lambert law of the form
\begin{align}
&\mathbb{I}(z,\widehat{\mathbf{e}}_{3})=\mathbb{E}\big\lbrace\widehat{\bm{\psi}(z,\widehat{\mathbf{e}}_{3})}\big\rbrace
=\psi_{o}\exp(-\mathsf{T}z)
\exp\left(\frac{1}{2}\alpha^{2}\mathsf{T}^{2}\bm{\mathsf{C}}\int_{0}^{z}\bigg|\int_{0}^{z_{1}}\mathlarger{\phi}(z_{1},z_{2};\xi))dz_{2}\bigg|dz_{1}\right)\nonumber\\&
=\psi_{o}\exp(-\mathsf{T}z)
\exp\left(\frac{1}{2}\alpha^{2}\mathsf{T}^{2}\bm{\mathsf{C}}\int_{0}^{z}\bigg|\int_{0}^{z_{1}}\exp(-|z_{1}-z_{2}|^{2}\xi^{-2})dz_{2}\bigg|dz_{1}\right)\nonumber\\&
=\psi_{o}\exp(-\mathsf{T}z)\exp\left(\frac{1}{4}\alpha^{2}\mathsf{T}^{2}\bm{\mathsf{C}}\xi\left[\exp(-z^{2}/\xi^{2})\left(\sqrt{\pi}z Erf\left(\frac{z}{\xi}\right)\exp\left(\frac{z^{2}}{\xi^{2}}\right)+\xi\right)-\xi\right]\right)
\end{align}
However, the light flux scattered from the beam will then be described by the RTE or its diffusion approximation if $\mathsf{S}\gg\mathsf{A}$. The random scatter coefficient will then affect the evolution and diffusion of the scattered light flux.(This may be considered in a followup article.)
\appendix
\renewcommand{\theequation}{\Alph{section}.\arabic{equation}}
\section{\textbf{Gaussian random functions}}
\begin{defn}(\textbf{Formal definition of Gaussian random fields})\newline
GRFS are formally summarized as follows. More details can be found in \textbf{[30-37]}.
Let $(\bm{\Omega},\mathcal{F},{\mathbb{P}})$ be a probability space. Within the probability triplet, $(\bm{\Omega},\mathcal{F})$ is a \textbf{measurable space}, where $\mathcal{F}$ is the $\sigma$-algebra (or Borel field) that should be interpreted as being comprised of all reasonable subsets of the state space $\bm{\Omega}$. Then:
\begin{enumerate}[(a)]
 \item ${{\mathbb{P}}}$ is a function such that ${{\mathbb{P}}}:\mathcal{F}\rightarrow [0,1]$, so that for all $\mathcal{B}\in\mathcal{F}$, there is an associated probability ${{\mathbb{P}}}(\mathcal{B})$. The measure is a probability measure when ${\mathbb{P}}(\bm{\Omega})=1$.
\item Let $\mathbf{x}_{i}\subset{{\mathbb{I\!D }}}\subset\bm{\mathbb{I\!R}}^{n}$ be Euclidean coordinates and let
$(\bm{\Omega},{\mathcal{F}},{{\mathbb{P}}})$ be a probability space. Let ${\mathbf{G}}(\mathbf{x};\omega)$ be a random scalar function that depends on the coordinates $\mathbf{x}\subset{\mathbb{I\!D}}\subset{\mathbb{R}}^{n}$ and also $\omega\in\bm{\Omega}$.
     \item Given any pair $(\mathbf{x},\omega)$ there $\bm{\exists}$ map $\mathcal{M}:{\mathbb{I\!R}}^{n}\times\bm{\Omega}\rightarrow{\mathbb{R}}$ such that
\begin{align}
\mathcal{M}:(\omega,\mathbf{x})\longrightarrow{\mathbf{G}}(\mathbf{x};\omega)
\end{align}
so that ${{\mathbf{G}}(\mathbf{x};\omega)}$ is a \textbf{random function or field} on $\mathbb{I\!D}\subset\mathbb{R}^{n}$ with respect to the probability space $(\bm{\Omega},\mathcal{F},{\mathlarger{\mathbb{P}}})$.
\item A random field is then essentially a family of random variables $\lbrace{\mathbf{G}}(\mathbf{x};\omega)\rbrace$ defined with respect to the space $(\bm{\Omega},\mathcal{F},{\mathbb{P}})$ and ${\mathbb{R}}^{n}$.
\item The fields can also include a time variable $t\in{\mathbb{ R}}^{+}$ so that given any triplet
$(\mathbf{x},t,\omega)$ there is a mapping $\mathfrak{M}:{\mathbb{R}}\times{\mathbb{R}}^{n}\times\bm{\Omega}\rightarrow {\mathbb{R}}$ such that $\mathfrak{M}:(\mathbf{x},t,\omega)\hookrightarrow \mathbf{G}(\mathbf{x},t;\omega)$ is a \textbf{spatio-temporal random field}.Normally, the field will be expressed in the form ${\mathbf{G}}(\mathbf{x},t)$ or ${\mathbf{G}}(\mathbf{x})$ with $\omega$ dropped.
\item The random field $\mathbf{G}(\mathbf{x})$ will have the following bounds and continuity properties:
\begin{align}
&{{\mathbb{P}}}\bigg[\sup_{\mathbf{x}\in\mathbb{I\!D}}|{\mathbf{G}}(\mathbf{x})|~~<~~\infty\bigg]~=+1\\&
{{\mathbb{P}}}\bigg[\lim_{\mathbf{x}\rightarrow\mathbf{y}}\big|{\mathbf{G}}(\mathbf{x})-{\mathbf{G}}(\mathbf{y})\big|=0,
~\forall~(\mathbf{x},\mathbf{y})\in\mathbb{I\!D}\bigg]=1
\end{align}
\end{enumerate}
\end{defn}
\begin{lem}
The random field is at the least, mean-square differentiable in that
\begin{align}
\nabla_{j}{\mathbf{G}}(\mathbf{x})=\frac{\partial}{\partial x_{j}}{\mathbf{G}}(\mathbf{x})= \lim_{\mathbf{h}\rightarrow 0} \frac{{\mathbf{G}}(\mathbf{x}+|\mathbf{h}|{\widehat{\mathbf{e}}}_{j})-{\mathbf{G}}(\mathbf{x})}{|\mathbf{h}|}
\end{align}
where $\mathlarger{\widehat{\mathbf{E}}}_{j}$ is a unit vector in the $j^{th}$ direction. For a Gaussian field, sufficient conditions for differentiability can be given in terms
of the covariance or correlation function, which must be regulated at $\mathbf{x}=\mathbf{y}$ The derivatives of the field $\nabla_{i}{\mathbf{G}}(\mathbf{x}),\nabla_{i}\nabla_{j}{\mathbf{G}}(\mathbf{x})$ exist at least up to 2nd order and do line, surface and volume integrals
$\mathlarger{\int}_{\bm{\Omega}}{\mathbf{G}}(\mathbf{x},t)d\mu(\mathbf{x})$ with respect to domain $\mathbb{I\!D}$.(See Appendix B.)
The derivatives or integrals of a random field are also a random field.
\end{lem}
\begin{defn}
The stochastic expectation $\mathbb{E}\lbrace\bullet\rbrace $ and binary correlation with respect to the space $(\Omega,{\mathcal{F}},{\mathlarger{\mathbb{P}}})$ is defined as follows, with $(\omega,\zeta)\in\mathbf{\Omega}$
\begin{align}
&\mathbb{E}\bigg\lbrace\bullet\bigg\rbrace=\mathlarger{\int}_{\omega}\bullet~d{{\mathbb{P}}}[\omega]\\&
\mathbb{E}\bigg\lbrace\bullet{\mathlarger{\otimes}}\bullet\bigg\rbrace=\mathlarger{\int}\!\!\!\!\mathlarger{\int}_{\Omega}
\bullet\mathlarger{\otimes}\bullet~d{{\mathbb{P}}}[\omega]
d{{\mathbb{P}}}[\zeta]
\end{align}
For Gaussian random fields $\bullet=\mathlarger{\mathscr{G}}(\mathbf{x})$ only the binary correlation is required so that
\begin{align}
&\mathbb{E}\bigg\lbrace{\mathbf{G}}(\mathbf{x})\bigg\rbrace=\mathlarger{\int}_{\omega}{\mathbf{G}}(\mathbf{x};\omega)~d\mathbb{P}[\omega]=0\\&
\mathbb{E}\bigg\lbrace{\mathbf{G}}(\mathbf{x})\mathlarger{\otimes}{\mathbf{G}}(\mathbf{y})\bigg\rbrace=
\mathlarger{\int}\!\!\!\!\mathlarger{\int}_{\Omega}{\mathbf{G}}(\mathbf{x};\omega)\mathlarger{\otimes}
{\mathbf{G}}(\mathbf{y};\xi)~d\mathbb{P}[\omega]d{{\mathbb{P}}}[\zeta]\nonumber\\&
=\mathlarger{\phi}(\mathbf{x},\mathbf{y};\xi)
\end{align}
and regulated at $\mathbf{x}=\mathbf{y}$ for all $(\mathbf{x},\mathbf{y})\in\mathbb{I\!D}$ if
\begin{align}
\mathbb{E}\big\lbrace{\mathbf{G}}(\mathbf{x})\mathlarger{\otimes}
{\mathbf{G}}(\mathbf{x})\big\rbrace=\bm{\mathsf{C}}<\infty
\end{align}
\end{defn}
\begin{defn}
The correlation between two or more fields is denoted by the operation $\mathlarger{\otimes}$. We say that two random fields $(\mathbf{G}(\mathbf{x}),\mathbf{G}(\mathbf{y}))$ defined for any $(\mathbf{x},\mathbf{y})\in\mathbb{I\!D}$ are correlated or uncorrelated if
\begin{empheq}[right=\empheqrbrace]{align}
&\mathbb{E}\big\lbrace{\mathbf{G}}(\mathbf{x})\mathlarger{\otimes}{\mathbf{G}}(\mathbf{y})\big\rbrace\ne 0 \nonumber\\&
\mathbb{E}\big\lbrace{\mathbf{G}}(\mathbf{x})\mathlarger{\otimes}{\mathbf{G}}(\mathbf{y})\big\rbrace= 0 \nonumber
\end{empheq}
\end{defn}
The binary 2-point function nor covariance fully determines all its properties, but the key advantages of GRSF are that a GRSF is Gaussian distributed and can be classified purely by its first and second moments, and all high-order moments and cumulants can be ignored. GRFs also tend to be more tractable. The same definitions can be applied to spatio-temporal fields.
\begin{defn}
If $\mathbf{G}(\mathbf{x})$ is a GRSF existing for all $\mathbf{x}\in\bm{\mathbb{R}}^{n}$ or $\mathbf{x}\in\mathbb{I\!D}\subset{\mathbb{R}}^{3}$ then
\begin{align}
&\mathbb{E}\big\lbrace{\mathbf{G}}(\mathbf{x})\big\rbrace=0\\&
\mathbb{E}\big\lbrace{\mathbf{G}}(\mathbf{x})\mathlarger{\otimes}~\mathbf{G}(\mathbf{y})
\bigg\rbrace=\phi(\mathbf{x},\mathbf{y};\xi)
\end{align}
for any 2 points $(\mathbf{x},\mathbf{y})\in\mathbb{I\!D}$, and with a correlation length $\lambda$. It is regulated if
$\mathbb{E}\lbrace({\mathbf{G}}(\mathbf{x})\mathlarger{\otimes}~{\mathbf{G}}(\mathbf{x})\big\rbrace=\bm{\mathsf{C}}<\infty$
For a white-in-space Gaussian noise or random field ${\mathbb{E}}\lbrace{\mathbf{G}}(\mathbf{x})\mathlarger{\otimes}~{\mathbf{G}}(\mathbf{y})\rbrace=
\bm{\mathsf{C}}\delta^{n}(\mathbf{x}-\mathbf{y})$ and is unregulated. This paper utilises only regulated GRSFs. The binary covariance is then
\begin{align}
&{\mathbb{COV}}\bigg\lbrace{\mathbf{G}}(\mathbf{x})\mathlarger{\otimes}~{\mathbf{G}}(\mathbf{y})\bigg\rbrace\equiv
{\mathbb{E}}\bigg\lbrace{\mathbf{G}}(\mathbf{x})\mathlarger{\otimes}~{\mathbf{G}}(\mathbf{y})\bigg\rbrace
+{\mathbb{E}}\bigg\lbrace{\mathbf{G}}(\mathbf{x})\bigg\rbrace\mathbb{E}\bigg\lbrace{\mathbf{G}}(\mathbf{y},t)\bigg\rbrace\nonumber\\&
={\mathbb{E}}\bigg\lbrace{\mathbf{G}}(\mathbf{x})\mathlarger{\otimes}~{\mathbf{G}}(\mathbf{y})\bigg\rbrace
=\phi(\mathbf{x},\mathbf{y};\xi)
\end{align}
\end{defn}
\section{\textbf{Stochastic Integration of random fields}}
The integral of a GRSF is defined as the limit of a Riemann sum of the field over the partition of a domain \textbf{34-37}
\begin{prop}
Let ${\mathbb{I\!D}}\subset{\mathbb{R}}^{n}$ be a (closed) domain with boundary $\partial{\mathbb{I\!D}}$ and $x=(x_{1},...,x_{n})\subset{\mathbb{I\!D}}$. Let $\mathbb{I\!D}=\bigcup_{q=1}^{M}\mathbb{I\!D}_{1}$ be a partition of $\mathbb{I\!D}$ with $\mathbb{I\!D}_{p}\bigcap\mathbb{I\!D}_{q}=\varnothing$ if $p\ne q$. Let $\mathbf{x}^{(q)}\in\mathbb{I\!D}_{q}$ for all $q=1...M$. Note $\mathbf{x}^{(q)}\equiv (x_{1}^{(q)},...x_{n}^{(q)})$ in $\mathbb{R}^{n}$. Then $\mathbf{x}^{(1)}\in\mathbb{I\!D}_{1}, \mathbf{x}^{(2)}\in\mathbb{I\!D}_{2},
...,\mathbf{x}^{(M)}\in\mathbb{I\!D}$. Let $\mathrm{Vol}(\mathbb{I\!D}_{q})$ be the volume of the partition $\mathbb{I\!D}_{q}$ so that $\mathrm{vol}(\mathbb{I\!D})=\sum_{q}^{M}\mathrm{Vol}(\mathbb{I\!D})$. Then let
\begin{equation}
\mathbb{I\!D}=\bigcup_{\xi=1}^{M}\mathbb{I\!D}_{q}
\end{equation}
be a partition of $\partial\mathbb{I\!D}$ with $\partial\mathbb{I\!D}_{\xi}$ be a partition of the boundary or surface into N constituents. Let $\bm{x}^{(q)}\in\partial\mathbb{I\!D}_{q}$ for all $q=1...M$. Note $\mathbf{x}^{(q)}\equiv (\mathbf{x}_{1}^{(q)},...\mathbf{x}_{n}^{(q)})$. Then $\mathbf{x}^{(1)}\in\partial\mathbb{I\!D}_{1}, \mathbf{x}^{(2)}\in\partial\mathbb{I\!D}_{2},...,\mathbf{x}^{(H)} \in\partial\mathbb{I\!D}$. Let $\|{\mathbb{D}}_{q}\|\equiv\mu(\partial\mathbb{I\!D}_{q})$ be the surface area of the partition $\partial\mathbb{I\!D}_{q}$ so that
\begin{align}
\partial\mathbb{I\!D}=\bigcap_{q=1}^{M}\partial\mathbb{I\!D}_{q} 
\end{align}
The total volume and area of $\mathbb{I\!D}$ is
\begin{align}
\mathrm{Vol}(\mathbb{I\!D})=\sum_{q=1}^{M}Vol(\mathbb{I\!D}_{q})
=\sum_{q=1}^{M}Vol(\mathbb{I\!D}_{q})
\end{align}
\begin{align}
\mathrm{Area}(\partial\mathbb{I\!D})=\sum_{q=1}^{M}\mathrm{Area}(\partial\mathbb{I\!D}_{q})
\end{align}
Given the probability triplet $(\Omega,\mathcal{F},{\mathbb{P}})$ then a Gaussian random field on $\mathbb{I\!D}$ for all $x\in\mathbb{I\!D}$ is ${\mathbf{G}}:\omega\times\mathbb{I\!D}\rightarrow {\mathbb{R}}$ and ${\mathbf{G}(x^{q},\omega)}\in{\mathbb{I\!D}}_{q}$ exists for all $x^{(q)}\in \mathbb{I\!D}_{q}$ and $\omega\in\Omega$. The stochastic volume integral and the stochastic surface integral are
\begin{align}
&\int_{\mathbb{I\!D}}{\mathbf{G}}{(\mathbf{x};\omega)}d\mu_{n}(\mathbf{x})=\lim_{all~v({\mathbb{I\!D}}_{q})\uparrow 0}\sum_{q=1}^{M}{\mathbf{G}}(\mathbf{x}^{(q)};\omega)\mathrm{Vol}(\mathbb{I\!D}_{q})\\&
\int_{\partial\mathbb{I\!D}}{\mathbf{G}}(\mathbf{x};\omega)d\mu_{n-1}(\mathbf{x})
=\lim_{all~\mathrm{Area}(\partial\mathbb{I\!D}_{q})\uparrow 0}\sum_{q=1}^{M}{\mathbf{G}}(\mathbf{x}^{(q)};\omega)Area(\partial\mathbb{I\!D}_{q})
\end{align}
When a Gaussian random field is integrated, it is the limit of a linear combination of Gaussian random variables/fields so it is again Gaussian.
\end{prop}
Next, the stochastic expectations or averages are defined.
\begin{prop}
Since
\begin{equation}
\mathbb{E}\bigg\lbrace\bullet\bigg\rbrace
=\int_{\mathbb{I\!D}}(\bullet)d{\mathbb{P}}(\omega)
\end{equation}
the expectation of the volume integral is as follows.
\begin{align}
&\mathbb{E}\bigg\lbrace\int_{\mathbf{Q}}{\mathbf{G}}(\mathbf{x};\omega)d\mu_{n}(\mathbf{x})\bigg\rbrace
\equiv\int\!\!\int_{\mathbb{I\!D}}{\mathbf{G}}(\mathbf{x};\omega)d\mu_{n}(x) d\mathbb{P}(\omega)\\&
=\lim_{M\uparrow\infty}\lim_{all~\mathrm{Vol}(\mathbb{I\!D}_{q})\uparrow 0}\int_{\Omega}\sum_{q=1}^{M}
{\mathbf{G}}(x^{(q)};\omega)\mathrm{Vol}(\mathbb{I\!D}_{q})d\bm{\mathbb{P}}(\omega)=0
\end{align}
which vanishes for GRSFs since $\mathbb{E}\big\lbrace{\mathbf{G}}(x^{(q)})\big\rbrace=0$. Similarly, for the stochastic surface integrals
\begin{align}
&\mathbb{E}\left\lbrace\int_{\partial{\mathbb{I\!D}}}{\mathbf{G}}(x;\omega)d\mu_{n-1}(x)
\right\rbrace\equiv\int_{\Omega}\int_{\partial{\mathbb{I\!D}}}{\mathbf{G}}(x;\omega)
d^{n-1}x d\bm{\mathbb{P}}(\omega)\\&=\lim_{all~\mu(\partial{\mathbb{I\!D}}_{q})\uparrow 0}\int_{\Omega}\sum_{q=1}^{M}{\mathbf{G}}(x^{(q)};\omega)
Area(\partial{\mathbb{I\!D}}_{q})d{\mathbb{P}}(\omega)=0
\end{align}
\end{prop}
In one dimension, along the z-axis, the stochastic integral is
\begin{align}
\int_{\mathbb{I\!D}}\mathbf{G}(z)dz=\int_{0}^{z}\mathbf{G}(\bar{z})d\bar{z}
\end{align}
Given an integral (or summation) over a random field or stochastic process, the Fubini theorem states that the expectation of the integral or sum over a random field is equivalent to the integral or sum of the expectation of the field.
\begin{thm}
Let ${\mathbf{G}}(x)$ be a random field not necessarily Gaussian, existing for all $x\in\mathbb{I\!D}$ with expectation $\mathbb{E}
\lbrace\bm{\mathbf{G}}(x)\rbrace $, not necessarily zero. Then
\begin{equation}
\mathbb{E}\bigg\lbrace\int_{\mathbb{I\!D}}{\mathbf{G}}(x)d\mu(x)
\bigg\rbrace\equiv \int_{{\mathbb{I\!D}}}\mathbb{E}\bigg\lbrace
{\mathbf{G}}(\mathbf{x})\bigg\rbrace d\mu_{n}(\mathbf{x})
\end{equation}
For a set of M random fields ${\mathbf{G}}_{q}(x)$
\begin{equation}
\mathbb{E}\bigg\lbrace\sum_{q=1}^{M}\bm{\mathbf{G}}_{q}(x)\bigg\rbrace=
\sum_{q=1}^{M}\mathbb{E}\bigg\lbrace\bm{\mathbf{G}}_{q}(x)\bigg\rbrace
\end{equation}
\end{thm}
}
\end{document}